\documentclass[letterpaper]{IEEEtran}

\usepackage[english]{babel}
\usepackage{amsmath,amssymb}
\usepackage{times}
\usepackage{relsize}
\usepackage{graphicx}
\usepackage{url}
\usepackage{booktabs}
\usepackage{multirow}
\usepackage{mparhack}
\usepackage{subfigure}
\usepackage{authblk}
\usepackage{amsthm}
\usepackage[noend]{algorithmic}
\usepackage[linesnumbered,ruled,vlined]{algorithm2e}

\usepackage{array}
\usepackage{cite}
\usepackage{eqparbox}
\usepackage{mdwmath}
\usepackage{balance}
\usepackage{epsfig}
\usepackage{xcolor}
\usepackage{xfrac}
\usepackage{mathtools}
\usepackage{bm}
\usepackage{amsfonts}
\usepackage[all]{xy}
\usepackage{etoolbox}
\usepackage{graphicx}
\DeclareMathAlphabet{\mathantt}{OT1}{antt}{li}{it}
\DeclareMathAlphabet{\mathpzc}{OT1}{pzc}{m}{it}

\usepackage{accents}
\newcommand\munderbar[1]{%
  \underaccent{\bar}{#1}}

\newtheorem{theorem}{Theorem}

\DeclareFontFamily{OT1}{pzc}{}
\DeclareFontShape{OT1}{pzc}{m}{it}%
  {<-> s * [1.1] pzcmi7t}{}
\DeclareMathAlphabet{\mathpzc}{OT1}{pzc}%
                     {m}{it}

\makeatletter
\patchcmd{\@maketitle}
  {\addvspace{0.75\baselineskip}\egroup}
  {\addvspace{-1.45\baselineskip}\egroup}
  {}
  {}
\makeatother

\title{Accounting for Information Freshness in Scheduling of Content Caching}
\author{Ghafour Ahani~and~\and Di Yuan}
\affil{Department of Information Technology, Uppsala University, Sweden\\
Emails:\{ghafour.ahani, di.yuan\}@it.uu.se}

\begin{document}

\maketitle
\begin{abstract}
In this paper, we study the problem of optimal scheduling of content placement along time in a base station with limited cache capacity, taking into account jointly the offloading effect and freshness of information. We model offloading based on  popularity in terms of the number of requests and information freshness based on the notion of age of information (AoI). The objective is to reduce the load of backhaul links as well as the AoI of contents in the cache via a joint cost function. For the resulting optimization problem, we prove its hardness via a reduction from the Partition problem. Next, via a mathematical reformulation, we derive a solution approach based on column generation and a tailored rounding mechanism. Finally, we provide performance evaluation results showing that our algorithm provides near-optimal solutions.
\end{abstract}
\IEEEpeerreviewmaketitle

\section{Introduction}
Content caching at the network edge is considered to be an enabler for future wireless networks. This technique strives to mitigate the heavy burden on backhaul links via providing the users with their contents of interest from the network edge without the need of going to the core networks.

In designing effective caching strategies, previous works have focused on content popularity, whereas another important aspect is information freshness. Popularity of a content is defined as the number of users requesting the content. Popularity may vary over time\cite{8327582}. Thus, some contents may be added to or removed from the cache as they become popular or unpopular. Freshness of contents in the cache refers to how recent the content has been obtained from the core network. The longer a content is stored in the cache without an update, the higher risk is that the cached content becomes obsolete. Hence, we would like to refresh  the cached contents often, which however leads to higher load on the backhaul. Freshness of contents naturally arises in applications such as news, traffic information, etc., and it may have a great impact on user satisfaction. We model freshness of contents using the notion of age of information (AoI). For content caching, AoI is defined as the amount of time elapsed since the time that the content is refreshed. In this paper, we use a joint cost function to address the trade-off between the benefit of offloading via caching and AoI.

The works such as \cite{7562037,Cost2018Deng,6883600,7414014} took into account only the popularities of contents in designing cache placement strategies. The works in \cite{7562037,Cost2018Deng} considered content caching with known popularities of contents. The studies in \cite{6883600,7414014}
 showed that the popularites of contents can be estimated via learning-based algorithms.
However, in the mentioned works popularity of a content are time-invariant. In \cite{8357917,Zhang2018Using}, caching with time-varying popularity profiles are investigated. In \cite{Zhang2018Using} an algorithm is proposed to estimated the time-varying popularities of contents.
The studies in \cite{8000687,Tang2019} considered information freshness but not popularity of contents in their caching problems. Recently, a few works \cite{8006505,8006506,8795490} have considered both popularity and freshness of contents. However, these works have the following limitations. In \cite{8006505}, the downloading cost of contents from the server is neglected. In \cite{8006506}, only one content of the cache could be updated in each time slot. In \cite{8795490}, it is assumed that the cache capacity is unlimited.

In this paper, we study optimal scheduling of content caching along time in a base station (BS) with limited storage capacity taking into account jointly offloading via caching and freshness of contents. The objective is to mitigate the load of backhaul links via minimizing a penalty cost function related to  content downloading, content updating, and AoI costs subject to the cache capacity. The main contributions of this work are summarized as follows:
\begin{itemize}
\item
The caching scheduling problem is formulated as an optimization problem. Specifically, it is formulated as an integer linear program (ILP) and the hardness of the problem is proved based on a reduction from the Partition problem.
\item
Via a problem reformulation, a column generation algorithm  (CGA) is developed. We prove that the subproblem of CGA can be converted to a shortest path problem that can be solved in polynomial time. In addition, the CGA provides an effective lower bound (LB) of global optimum .
\item
The solution obtained from CGA could be fractional, thus an advanced and problem-tailored rounding algorithm (RA) is derived to construct integer solutions.

\item
Simulations show the effectiveness of our solution approach by comparing the obtained solutions to the LB as well as the conventional algorithms. Our algorithm provides solutions within $1\%$ of global optimum.
\end{itemize}

\section{System Scenario and Problem Formulation}

\subsection{System Scenario}\label{System_Scenario}
 The system scenario consists of a content server, a BS and a set of users $\mathcal{U}=\{1,2,\dots,U\}$  within the coverage of the BS. The server has all the contents, and the BS is equipped with a cache device of capacity $S$.
The contents are dynamic, i.e., the information they contain may change over time. Denote by $\mathcal{F}=\{1,2,\dots,F\}$ the set of the contents. We assume the server has always the up-to-date version of the contents.
Denote by $l_f$ the size of content $f$.
Each content is either fully stored or not stored at all at the BS. The system scenario is shown in Figure \ref{SystemScenario}.

\begin{figure}[ht!]
\centering
\includegraphics[scale=0.40]{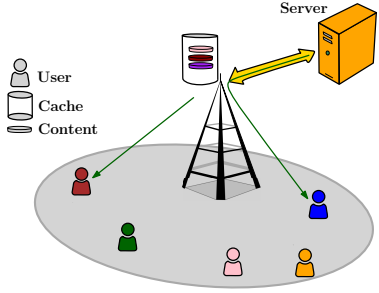}
\begin{center}
\caption{System scenario.}
\label{SystemScenario}
\end{center}
\end{figure}
We consider a slotted time system of $\mathcal{T}=\{1,2,\dots,T\}$ time slots. At the beginning of each time slot, the contents to be stored in the cache need to be determined by an updating/placement action. Namely, some stored contents may be removed from the cache, some contents may be added to the cache, and some contents may be re-downloaded from the server. The freshness of a content may decrease along time. We use AoI to model the freshness of contents. A content that is newly downloaded from the cache has AoI $0$, and for each time slot it remains in the cache without re-downloading, its AoI increases with one time slot. Denote by $p_f(i)$ the cost associated with an AoI of $i$ time slots for content $f$. A content has AoI $i$ time slots when the content has been stored in the cache for $i$ continuously time slots without any update.

In our model, user $u\in \mathcal{U}$ requests at most $R_u$ contents within the $T$ time slots based on its interest. The set of requests for user $u$ is denoted by $\mathcal{R}_u=\{1,\dots,R_u\}$. The downloading process of a content starts as soon as the request is made. The content can be downloaded either from the cache if the content is in the cache, or otherwise from the server. We assume the time of each request is known or can be predicted via using a prediction model, e.g., the one in \cite{Zhang2018Using}. For user $u$ and its $r$-th request, the requested content and the time slot of request are denoted by $h(u,r)$ and $o(u,r)$, respectively.

\subsection{Cost Model}
Denote by $x_{tf}$ a binary optimization variable which equals one if and only if the $f$-th content is stored in time slot $t$. Denote by $c_s$ and $c_b$ the costs for downloading one unit of data from the server and the cache to a user, respectively. We have $c_s>c_b$ to encourage downloading from the cache.
The downloading cost for user $u$ to obtain its $r$-th request, denoted by $C_{ur}$, is expressed as:
\begin{equation}
\begin{aligned}
C_{ur}=l_{h(u,r)}[c_bx_{o(u,r)h(u,r)}+c_s(1-x_{o(u,r)h(u,r)})].
\end{aligned}
\end{equation}
The downloading cost for completing all requests of all users, denoted by $C_{download}$, is $C_{download}=\sum_{u=1}^{U}\sum_{r=1}^{R_u}C_{ur}$.

 Denote by binary variable $a_{tfi}$, $i \in \{0,1,\dots,t-1\}$, whether or not content $f$ is in the cache and has AoI $i$ time slots. The overall AoI cost is expressed as:
\begin{equation}
\begin{aligned}
C_{AoI}=\sum_{f=1}^{F}\sum_{t=1}^{T}\sum_{i=1}^{t-1}p_f(i)m_{tf}a_{tfi},
\label{agecostf}
\end{aligned}
\end{equation}
where $m_{tf}$ is the number of users requesting content $f$ in time slot $t$. Updating contents in the cache incurs an updating cost.
The updating cost, denoted by $C_{update}$, is expressed as:
\begin{equation}
\begin{aligned}
C_{update}=\sum_{t=1}^{T}\sum_{f=1}^{F}l_f(c_s-c_b)a_{tf0},
\label{c_update}
\end{aligned}
\end{equation}
where $a_{tf0}$ means that the content is just downloaded from the server and has cost $l_f(c_s-c_b)$. Here $(c_s-c_b)$ is the downloading cost unit from the server to the cache. Finally, the total cost is denoted by $C_{total}$ and expressed as:
\begin{equation}
\begin{aligned}
C_{total}=C_{dowlad}+C_{update}+\lambda C_{AoI}.
\label{agecostf}
\end{aligned}
\end{equation}
Here, $\lambda$ is a weighting factor between $C_{AoI}$ and $C_{update}$. Larger $\lambda$ means frequently updating the contents of the cache and consequently smaller AoI for cached contents.
\subsection{Problem Formulation}
The update-enabled caching problem (UECP) is formulated as an ILP, and shown in (\ref{UECP}).
\vspace{-5mm}
\begin{figure}[!h]
\begin{subequations}
\begin{alignat}{2}
\text{(UECP)}~~ &\min\limits_{\bm{x},\bm{a}}\quad  C_{download}+C_{update}+\lambda C_{AoI}\\
\text{s.t}. \quad
& \sum_{f=1}^{F}x_{tf}l_f \leq S,t\in \mathcal{T}, \label{const:cacheSize}\\
& \sum_{i=0}^{t-1}a_{tfi}=x_{tf} ,t\in \mathcal{T}, f \in\mathcal{F}, \label{const:a1}\\
& a_{tfi}\ge x_{tf}+a_{(t-1)f(i-1)}-a_{tf0}-1 ,
\nonumber \\&t\in \mathcal{T}\setminus\{1\},f \in \mathcal{F},i\in\{1,\dots,t-1\}, \label{const:a2}\\
& x_{tf},a_{tfi}\in \{0,1\},t\in \mathcal{T},f\in \mathcal{F},
i \in \{0,\dots,t-1\}.
\end{alignat}
\label{UECP}
\end{subequations}
\end{figure}
\vspace{-7mm}
Constraints~(\ref{const:cacheSize}) indicate that used storage space is less than or equal to the cache capacity in each time slot.
Constraints~\eqref{const:a1} state that if the content is in the cache, it has to have one of the AoIs $0,\dots,t-1$.
Constraints~\eqref{const:a2} indicate content $f$ in time slot $t$ has AoI $i$ if and only if the content is in the cache in time slot $t$, has not AoI $0$ in time slot $t$, and has AoI $i-1$ in time slot $t-1$.

Even though this ILP can be solved by a standard solver, it needs significant computational time. Exploiting the structure of the problem, we develop an solution method based on column generation.

\subsection{Complexity Analysis}
\begin{theorem}
\textbf{UECP} is $\mathcal{NP}$-hard.
\end{theorem}
\begin{proof}
The proof is established by a
polynomial reduction from the Partition problem that is $\mathcal{NP}$-complete \cite{garey1979computers}.
Consider a Partition problem with a set of $N$ integers, i.e., $\mathcal{N}=\{n_1,\dots,n_N\}$. The task is to decide whether it is possible to partition $\mathcal{N}$ into two subsets $\mathcal{N}_1$ and $\mathcal{N}_2$ with equal sum.

The reduction is constructed as follows.
We set the cache capacity as $S=\frac{1}{2}\sum_{i=1}^{N}n_i$, the set of contents to $\mathcal{F}=\{1,\dots,N\}$, size of content $f \in \mathcal{F}$ to $l_f=n_f$, and the number of time slots to one, i.e., $T=1$. As $T=1$, there is no updating or AoI costs. The time slots of all requests are set to $1$, i.e., $o(u,r)=1, u \in \mathcal{U}, r \in \mathcal{R}_u$. We set $m_{1f}=2$ for $f \in \mathcal{F}$, $c_s=2$, and $c_b=1$. By this setting, if the cache stores content $f$, $4l_f-l_f-2l_f=l_f$ gain is achieved. As the cache capacity is $S=\frac{1}{2}\sum_{i=1}^{N}n_i$, a maximum possible of $\frac{1}{2}\sum_{i=1}^{N}n_i$ gain can be achieved. Now, the question is if this maximum gain can be achieved. This question can be answered by solving UECP which also will answer the Partition problem. Hence the conclusion.
\end{proof}

\section{Reformulation of UECP}
We provide a reformulation of the problem that enables a CGA. We define the caching and updating decisions for content $f$ across the $T$ time slots as tuple $(\bm{x}_f,\bm{a}_f)$ in which
$\bm{x}_f=[x_{1f}, x_{2f},\dots,x_{Tf}]^ \mathrm{ T }$ and $\bm{a}_f=[a_{1f0}, a_{2f0},\dots,a_{Tf(t-1)}]^ \mathrm{ T }$.
In total, $3^T$ of such tuples exist and one of them is used in a solution. Denote by $\mathcal{K}=\{1,2,\dots,3^T\}$  the index set for all possible solutions. We refer to a possible solution as a column. The cost of column $k\in\mathcal{K}$ for content $f\in \mathcal{F}$ is denoted by $C_{fk}$ and can be calculated by the formula in (\ref{Cfk}).
\begin{equation}
\begin{aligned}
C_{fk}&=\sum_{t=1}^{T} l_{f}m_{tf}[c_bx_{tf}^{(k)}+c_s(1-x_{tf}^{(k)})]\\
&+\sum_{t=1}^{T}l_f(c_s-c_b)a^{(k)}_{tf0}+\lambda\sum_{t=1}^{T}\sum_{i=1}^{t-1}p_f(i)m_{tf}a_{tfi}^{(k)}.\label{Cfk}
\end{aligned}
\end{equation}
In~(\ref{Cfk}), $x_{tf}^{(k)}$ and $a^{(k)}_{tfi}$ are constants and represent the values of $x_{tf}$ and $a_{tfi}$ with respect to $k$-th column, respectively.
Now, ILP~\eqref{UECP} can be reformulated as \eqref{PR}.
\begin{figure}[!h]
\vskip -5pt
\begin{subequations}
\begin{alignat}{2}
~~~~~~~~~~&\min\limits_{\bm{w}}\quad  \sum_{f\in \mathcal{F}}\sum_{k\in \mathcal{K}}C_{fk}w_{fk} \label{MPC1}\\
\text{s.t}. \quad
&  \sum_{f\in \mathcal{F}}\sum_{k\in \mathcal{K}}l_fx^{(k)}_{tf}w_{fk} \leq S,t\in \mathcal{T} \label{MP_C1}\\
&\sum_{k\in \mathcal{K}}w_{fk}=1,f\in \mathcal{F}\\
& w_{fk}\in \{0,1 \},f\in \mathcal{F},k\in \mathcal{K}. \label{MP_C2}
\end{alignat}
\label{PR}
\vskip -20pt
\end{subequations}
\end{figure}

Here, $w_{fk}$ is a binary variable where $w_{fk}=1$ if and only if the $k$-th column of content $f$ is selected, otherwise it is zero. Constraints \eqref{MP_C1} are the cache capacity constraints, and constraints \eqref{MP_C2} indicate that only one of the columns is used.

\section{Algorithm Design}\label{alg_design}
In this section, we present our solution method which consists of two algorithms. Algorithm $1$ is a column generation algorithm (CGA) applied to the continuous version of \eqref{PR}. Algorithm $2$ is a rounding algorithm (RA) applied to the solution obtained from CGA if the solution is fractional. These algorithms are applied alternately until an integer solution is constructed.
The solution method is shown in Algorithm~$\ref{alg_CGAandERA}$.  The term RMP in the algorithm will be discussed later.

\begin{algorithm}\label{alg_CGAandERA}
\caption{CGA and RA}
\begin{algorithmic}[1]
\algsetup{linenosize=\tiny}
\small
\STATE STOP $\leftarrow 0$
\WHILE{(STOP$=0$)}
\STATE Apply CGA to RMP and obtain $\bm{w}^*$
\IF {($\bm{w}^*$ is an integer solution)}
\STATE STOP $\leftarrow 1$
\ELSE
\STATE Apply RA to $\bm{w}^*$
\ENDIF
\ENDWHILE
\end{algorithmic}
\end{algorithm}

\subsection{Column Generation Algorithm}
In column generation, the problem is decomposed into a so called master problem (MP) and a subproblem (SP). The algorithm starts with a subset of columns and solves alternately MP and SP.
Each time SP is solved a new column that possibly improves the objective function is generated. The benefit of CGA is to exploit the fact that at optimum only a few columns are used.

\subsubsection{MP and RMP}
MP is the continuous version of formulation \eqref{PR}. Restricted MP (RMP) is the MP but with a small subset $\mathcal{K}^\prime_f\subset\mathcal{K}$ for any content $f\in \mathcal{F}$. RMP is expressed in \eqref{RMP}. Denote by $K^\prime_f$ the cardinality of $\mathcal{K}^\prime_f$.
\begin{figure}[!h]
\vskip -5pt
\begin{subequations}
\begin{alignat}{2}
\text{(RMP)}~~~~~~~~~~&
\min\limits_{\bm{w}}\quad  \sum_{f\in \mathcal{F}}\sum_{k\in \mathcal{K}^\prime_f}C_{fk}w_{fk}
 \label{obj:RMP} \\
\text{s.t}. \quad
&  \sum_{f\in \mathcal{F}}\sum_{k\in \mathcal{K}^\prime_f}l_fx^{(k)}_{tf}w_{fk} \leq S,t\in \mathcal{T}, \label{RMP_cachecapa}\\
&\sum_{k\in \mathcal{K}^\prime_f}w_{fk} = 1,f\in \mathcal{F},\label{RMP_1col}\\
& 0\le w_{fk} \le 1,f\in \mathcal{F},k\in \mathcal{K}^\prime_f.
\end{alignat}
\label{RMP}
\vskip -20pt
\end{subequations}
\end{figure}

\subsubsection{Subproblem}
The SP uses the dual information to generate new columns.
Denote by $\mathbf{w}^*=\{w^*_{fk}, f\in \mathcal{F} ~\text{and}~ k\in \mathcal{K}^\prime_f\}$ the optimal solution of RMP.
Denote by $\bm{\pi}^*$ and $\bm{\beta}^*$ the corresponding optimal dual variables of \eqref{RMP_cachecapa} and \eqref{RMP_1col}, respectively, i.e.,
$\bm{\pi}^*=[\pi^*_1,\pi^*_2,\dots,\pi^*_{T}]^\mathrm{ T }$ and $\bm{\beta}^*=[\beta^*_1,\beta^*_2,\dots,\beta^*_F]^ \mathrm{ T }$.
After obtaining $\mathbf{w}^*$, we need to check whether $\mathbf{w}^*$ is the optimal solution of RMP.
This can be determined by finding a column with the minimum reduced cost for each content $f\in \mathcal{F}$. If all these values are nonnegative, the current solution is optimal. Otherwise, we add the columns with negative reduced cost to  corresponding sets.

Given $\bm{\pi}^*$ and $\bm{\beta}^*$ for content $f\in \mathcal{F}$, the reduced cost of column $(\bm{x}_f,\bm{a}_f)$ is $C_{f}-\sum_{t=1}^{T}l_f\pi^*_t x_{tf}-\beta^*_f$ where $C_{f}$ can be computed using expression~(\ref{Cfk}) in which constants $x_{tf}^{(k)}$ and $a_{tfi}^{(k)}$ are replaced with optimization variables $x_{tf}$ and $a_{tfi}$, respectively.
To find the column with minimum reduced cost for content $f\in \mathcal{F}$, we need to solve subproblem SP$_f$, shown in \eqref{SP}. Denote by $(\bm{x}^*_f,\bm{a}^*_f)$ the optimal solution of SP$_f$.
If the reduced cost of $(\bm{x}^*_f,\bm{a}^*_f)$ is negative, we add it to $\mathcal{K}^\prime_f$.

\begin{figure}[!h]
\vskip -5pt
\begin{subequations}
\begin{alignat}{2}
(\text{SP}_f)~~~~~~~&
\min\limits_{(\bm{x_f},\bm{a_f})}\quad  C_{f}-\sum_{t=1}^{T}l_f\pi^*_t x_{tf}-\beta^*_f \label{SP_objective}\\
\text{s.t}. \quad
& \sum_{i=0}^{t-1}a_{tfi}=x_{tf} ,t\in \mathcal{T}, f \in\mathcal{F}, \label{constSP:a1}\\
& a_{tfi}\ge x_{tf}+a_{(t-1)f(i-1)}-a_{tf0}-1 ,\nonumber \\
& t\in \mathcal{T}\setminus\{1\},f \in \mathcal{F},i \in\{1,\dots,t-1\}, \label{constSP:a2}\\
&a_{tfi}\le a_{(t-1)f(i-1)},t\in \mathcal{T}\setminus\{1\},f \in \mathcal{F},\nonumber \\
&i \in\{1,\dots,t-1\}, \label{constSP:a3}\\
& x_{tf} \in \{0,1\},t\in \mathcal{T},f\in \mathcal{F},\\
&a_{tfi}\in \{0,1\},t\in \mathcal{T},f\in \mathcal{F},i \in \{0,\dots,t-1\}.
\end{alignat}
\label{SP}
\vskip -20pt
\end{subequations}
\end{figure}
Even though \eqref{SP} is an ILP, in the following, we show that it can be solved as a shortest path problem using for example Dijkstra's algorithm\cite{Cormen2009introduction} in polynomial time.
\begin{theorem}
For content $f \in \mathcal{F}$, SP$_f$ can be solved in polynomial time as a shortest path problem.
\label{shortest_path}
\end{theorem}
\begin{proof}
Consider content $f\in\mathcal{F}$. We construct an acyclic directed graph where finding the shortest path from the source to distention is equivalent to solving SP\text{$_f$}. The objective function~\eqref{SP_objective} can be rewritten as  \eqref{re_obj}. Denote by $C$ the total cost for downloading content $f$ via the server for all users requesting the content over all time slots, i.e., $C=\sum_{t=1}^{T}l_fm_{tf}c_s$. Denote by $v_{it}=p_f(i)m_{tf}$ the scenario where $m_{tf}$ users request content $f$ in time slot $t$ and the content has AoI $i$.
Denote by $c_1=l_f(c_s-c_b)$  the downloading cost from the server to the cache. Denote by $g_t=l_fm_{tf}(c_s-c_b)-l_f\pi^*_t$ the reduction in $C$ due to storing content $f$.

The graph is constricted as follows. Nodes $S$ and $D$ are used to represent the source and destination. Node $V_{00}$ is used to represent $x_0=0$. For time slot $t$, there are $t+1$ vertically aligned nodes. Using node $V_{t0}$ means that the content is not in the cache, and using node $V_{t1}^i$, $i \in \{0,\dots,t-1\}$, means that the content is in the cache and has AoI $i$. From node $S$ to $V_{00}$ there is an arc with weight $C$. For each node $V_{t0}$ there are two outgoing arcs one to $V_{(t+1)0}$ which means that the content is not stored in the next time slot and has weight $0$, and the other to $V_{(t+1)1}^{0}$ which has weight $c_1-g_t$ and means that the content is downloaded to the cache in the next time slot and has AoI $0$. For each node $V_{t1}^i$  there three outgoing arcs to $V_{(t+1)0}$, $V_{(t+1)1}^0$, and  $V_{(t+1)1}^{(i+1)}$, respectively. Using the first arc means that the content is deleted for the next time slot and has weight $0$. Using the second arc means the content is re-downloaded from the cache and has AoI~$0$ with weight $c_1-g_t$. Using the third arc means that the content is kept and its AoI increases with one unit and has weight $v_{(i+1)(t+1)}-g_{(t+1)}$. Finally, there are $T$ arcs from $V_{T0}$ and $V_{T1}^{i}$ for $i\in\{0,\dots,T-1\}$ to $D$ each with weight $-\beta_f$.

 Given any solution of \eqref{SP}, by construction of the graph, the solution directly maps to a path from the source to the destination with the same objective function.
 Conversely, given a path we construct an ILP solution. For time slot $t$, if flow is in node $V_{t0}$ then we set $x_{tf}=0$. If the flow is in $V_{t1}^i$, we set $x_{t1}=1$ and $a_{tfi}=1$. The resulting ILP solution has the same objective function value as length of the given path in terms of the arcs's weights. Hence the conclusion.
\end{proof}

\begin{figure*}[!t]
\normalsize
\begin{equation}
\label{re_obj}
\begin{aligned}
&C_{f}-\sum_{t=1}^{T}l\pi^*_t x_{tf}=\sum_{t=1}^{T} l_fm_{tf}[c_bx_{tf}+c_s(1-x_{tf})]
+\sum_{t=1}^{T}l_f(c_s-c_b)a_{tf0}+\sum_{t=1}^{T}\sum_{i=1}^{t-1}p_f(i)m_{tf}a_{tfi}-\sum_{t=1}^{T}l_f\pi^*_{tf} x_{tf}\\
&~~~~~~~~~~~~~~~~~~~~=\underbrace{\sum_{t=1}^{T}l_fm_{tf}c_s}_{C}+
\sum_{t=1}^{T}\left[\underbrace{l_f(c_s-c_b)}_{c_1}a_{tf0}+\sum_{i=1}^{t-1}\underbrace{p_f(i)m_{tf}}_{v_{it}}a_{tfi}- \underbrace{\left[l_fm_{tf}(c_s-c_b)-l_f\pi^*_{tf} \right]}_{g_t}x_{tf}\right].
\end{aligned}
\end{equation}
\hrulefill
\vspace*{-10pt}
\end{figure*}
\begin{figure*}
\centering
\includegraphics[scale=0.35]{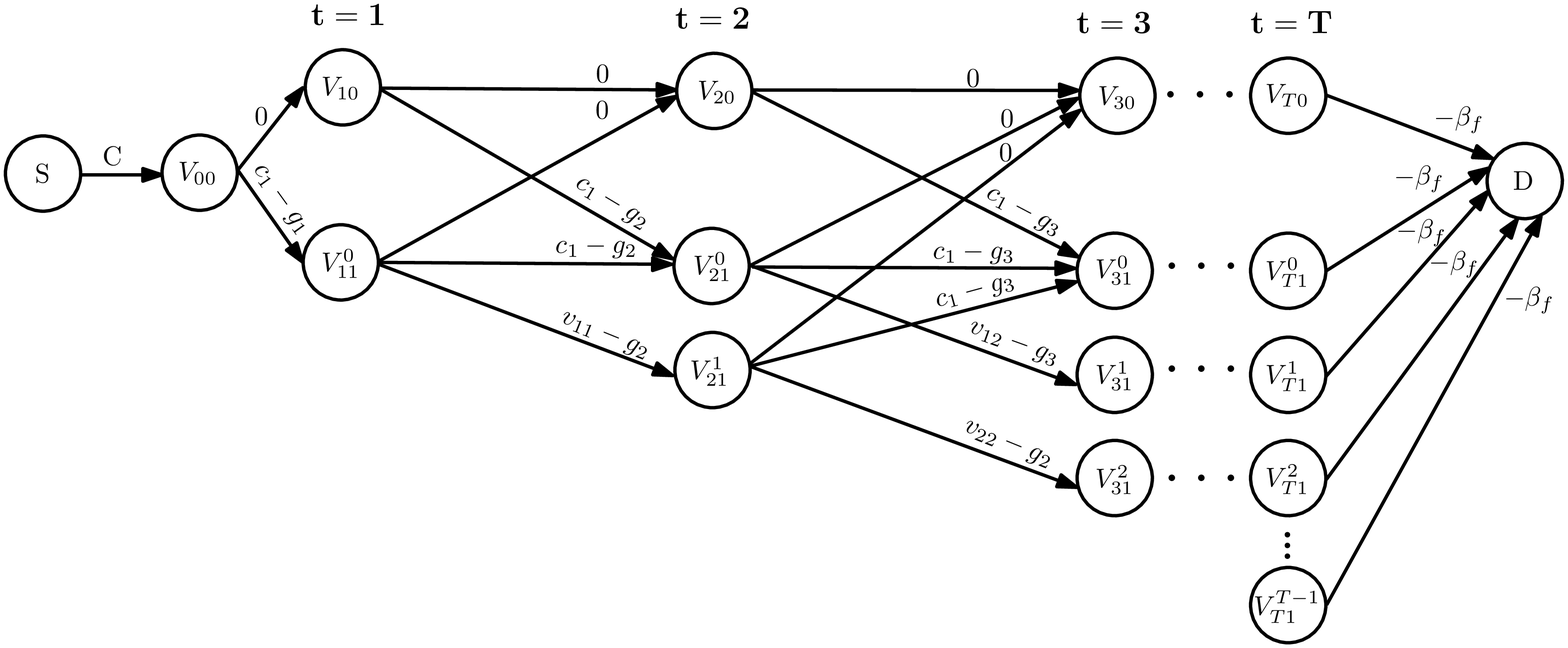}
\begin{center}
\vspace{-10mm}
\caption{Graph of the shortest path problem for subproblem.}
\label{SP1}
\end{center}
\vspace{-8mm}
\hrulefill
\end{figure*}

\begin{algorithm}\label{alg_CGA}
\caption{Column Generation Algorithm (CGA)}
\begin{algorithmic}[1]
\algsetup{linenosize=\tiny}
\small
\REQUIRE $S$, $c_b$, $c_s$, $\lambda$, $l_f$, $f \in \mathcal{F}$, $o(u,r)$ and $h(u,r)$, u $\in \mathcal{U},$ \\$r\in \mathcal{R}_u$
\ENSURE $\mathbf{w}^*$
\STATE $\mathcal{K}^\prime_f \leftarrow \{(\bf{0}^\mathrm{T},\bf{0}^\mathrm{T})\}$ for $f \in \mathcal{F}$
\STATE STOP $\leftarrow 0$
\WHILE{(STOP$=0$)}
\STATE Solve RMP and obtain $\mathbf{w}^*$, $\bm{\pi}^*$, and $\bm{\beta}^*$
\STATE STOP $\leftarrow 1$
\FOR{$f=1$ ~to~ $F$}
\STATE Solve SP$_f$ and obtain $(\bm{x}^*_f,\bm{a}^*_f)$
\IF {$C_{fk^*}-\sum_{t=1}^{T}l_f\pi^*_t x^*_{tf}-\beta^*_f <0$}
\STATE $\mathcal{K}^\prime_{f} \leftarrow \mathcal{K}^\prime_{f}\cup \{(\bm{x}^*_f,\bm{a}^*_f)\}$
\STATE STOP $\leftarrow 0$
\ENDIF
\ENDFOR
\ENDWHILE
\end{algorithmic}
\end{algorithm}

\subsection{Rounding Algorithm}
The solution of CGA could be fractional. Thus, we need a mechanism to construct integer solutions. We design a rounding algorithm (RA) to achieve this. RA repeatedly fixes the caching decisions of contents over time slots until an integer solution is constructed. The caching decision for content $f$ and time slot $t$ is determined based on value $z_{tf}$, defined as $z_{tf}=\sum_{k\in \mathcal{K}^\prime_f}x^{(k)}_{tf}w^*_{fk}$. This value indicates how likely it is optimal to store content $f$ in time slot $t$. In the following we prove a relationship between $\bm{z}$ and $\bm{w}$ and then give the RA.

\begin{theorem}
For any content $f\in \mathcal{F}$, $w^*_{fk}$ is binary for any $k$ if and only if every element of $\bm{z}_{f}=[z_{1f},z_{2f},\dots,z_{Tf}]$ is binary.
\label{IntegerTheory}
\end{theorem}

\begin{proof}
For necessity, for any content $f\in \mathcal{F}$, if $w^*_{fk}$ is binary for any $k$, $k\in \mathcal{K}^\prime_f$, it is obvious from the definition that all elements of $\bm{z}_{f}$ are binary.
Now, we prove the sufficiency.
For any content $f\in \mathcal{F}$, assume that every element in $\bm{z}_{f}$ is binary. Assume that
$w_{fk}^*>0$ for $k \in K_f^{\prime\prime}\subseteq K_{f}^\prime$, then $z_{tf}=\sum_{k\in \mathcal{K}_f^{\prime\prime}}x^{(k)}_{tf}w^*_{fk}$.
To satisfy that element $\bm{z}_{tf}$ for $t \in \mathcal{T}$ is binary, elements $x_{tf}^{(k)}$ for $k \in \mathcal{K}^{\prime\prime}$  either are all zero or all one. Otherwise, as $\sum_{k\in\mathcal{K}^{\prime\prime}}w^*_{fk}=1$, one of the $z_{tf}$ becomes fractional. This means that all columns corresponding to $w^*_{fk}$ for $k \in K_f^{\prime\prime}$ must be the same.  Having two vectors with the same values violates the condition that the columns of any two $w^*_{fk}$ are different. Therefore, for any content $f$, $f\in \mathcal{F}$, if $z_{tf}$ is binary for any $t$, $t\in \mathcal{T}$, then $w^*_{fk}$ is an binary for any $k$, $k\in \mathcal{K}^\prime_f$. Hence the proof.
\end{proof}

RA consists of three main steps which are shown in Algorithm~\ref{alg_ERA}. First, for content $f \in \mathcal{F}$ in time slot $t\in \mathcal{T}$, the decision is to store the content if $z_{tf}=1$. All columns that do not comply with this caching decision will be discarded. These are done by Lines $2$-$3$. Second, the element of $\bm{z}$ being closest to zero or one is found and rounded. Based on the rounding outcome, the caching decision is determined and non-complying columns are discarded. These are done via Lines~$4$-$16$. Finally, the algorithm fixes the decisions of the contents across the time slots to zero if there is no remained spare space in the cache to store them in these time slots. This is done by Lines~$20$-$23$. The caching decisions made until now will be remained fixed in all subsequent iterations. Note that with these fixings SP\text{$_f$} can still be solved as a shortest problem. If $x_{tf}$ is set to $0$, nodes $V_{t1}^i$ for $i\in\{0,\dots,t-1\}$ and their connected arcs will be deleted from the graph. If $x_{tf}$ is set to $1$, node $V_{t0}$ and its connected arcs will be deleted.
\vspace{-1mm}
\begin{algorithm}\label{alg_ERA}
\caption{Rounding Algorithm (RA)}
\begin{algorithmic}[1]
\algsetup{linenosize=\tiny}
\small
\REQUIRE $\bm{w}^*$ and $(\bm{x},\bm{a})$
\STATE Compute $\bm{z}=\{z_{tf}, t \in \mathcal{T}, f \in \mathcal{F}\}$ where $z_{tf}=\sum_{k\in \mathcal{K}^\prime_f}x^{(k)}_{tf}w^*_{fk}$
\STATE  Fix $x_{tf}=1$ in SP$_f$ if $z_{tf}=1$, $t \in \mathcal{T}, f \in \mathcal{F}$ \label{fixxto1}
\STATE  Fix $w_{fk}=0$ in RMP if $x^{(k)}_{tf}=0$, $k \in \mathcal{K}^\prime_f, t \in \mathcal{T}, f \in \mathcal{F}$\label{fixyto0}
\STATE $\munderbar{z}\leftarrow\underset{t\in\mathcal{T}, f\in \mathcal{F}~~~~~~~~~~~~~~~~~~~~~~~~~~~} {\min\{z_{tf}| z_{tf}>0~\text{and}~z_{tf}<1\}}$\label{minz}
\STATE $(\munderbar{t},\munderbar{f})\leftarrow\underset{t\in\mathcal{T}, f\in \mathcal{F}~~~~~~~~~~~~~~~~~~~~~~~~~~~~~~~} {\arg\min\{z_{tf}| z_{tf}>0~\text{and}~z_{tf}<1\}}$\label{minzloc}

\STATE $\bar{z}\leftarrow\underset{t\in\mathcal{T}, f\in \mathcal{F}~~~~~~~~~~~~~~~~~~~~~~~~~~~~~~~~} {\min\{1-z_{tf}| z_{tf}>0~\text{and}~z_{tf}<1\}}$\label{maxz}
\STATE $(\bar{t},\bar{f})\leftarrow\underset{t\in\mathcal{T}, f\in \mathcal{F}~~~~~~~~~~~~~~~~~~~~~~~~~~~~~~~~~~~~} {\arg\min\{1-z_{tf}| z_{tf}>0~\text{and}~z_{tf}<1\}}$\label{maxzloc}

\IF {$(\munderbar{z}<\bar{z})$} \label{nearestcheck}
\STATE Fix $x_{\munderbar{t}\munderbar{f}}=0$  in SP$_{\munderbar{f}}$\label{fixxunderbarto0}
\STATE  Fix $w_{\munderbar{f}k}=0$ if $x^{(k)}_{\munderbar{t}\munderbar{f}}=1$, $k \in \mathcal{K}^\prime_{\munderbar{f}}$\label{fixyunderbarto0}

\ELSIF{$(l_{\bar{f}}\le S^\prime)$}
\STATE Fix $x_{\bar{t}\bar{f}}=1$ in SP$_{\bar{f}}$\label{fixxbarto1}
\STATE  Fix $w_{\bar{f}k}=0$ if $x^{(k)}_{\bar{t}\bar{f}}=0$, $k \in \mathcal{K}^\prime_{\bar{f}}$\label{fixybarto0}
\ELSE
\STATE Fix $x_{\bar{t}\bar{f}}=0$  in SP$_{\bar{f}}$\label{fixxbarto0}
\STATE  Fix $w_{\bar{f}k}=0$ if $x^{(k)}_{\bar{t}\bar{f}}=1$, $k \in \mathcal{K}^\prime_{\bar{f}}$\label{fixybarto01}
\ENDIF
\FOR{$t=1$ ~to~ $T$}
\STATE $\mathcal{F}^\prime \leftarrow \{f \in \mathcal{F}| x_{tf} \text{ is set to one}\}$
\STATE $S^\prime\leftarrow S-\sum_{f \in \mathcal{F}^\prime}l_f$
\FOR{$f \in \mathcal{F}\backslash \mathcal{F}^\prime$}
\IF {$l_f> S^\prime$}\label{fixto0bysize1}
\STATE Fix $x_{tf}=0$ in SP$_{f}$
\STATE  Fix $w_{fk}=0$ in RMP if $x^{(k)}_{tf}=1$, $k \in \mathcal{K}^\prime_{f}$ \label{fixto0bysize2}
\ENDIF
\ENDFOR
\ENDFOR
\end{algorithmic}
\end{algorithm}

\section{Performance Evaluation}\label{sec:performance}
We compare CGA to the LB and two conventional caching algorithms: random-based algorithm (RBA) \cite{7959865} and popularity-based algorithm (PBA) \cite{Ahlehagh2014Video}.
Both algorithms treat contents one by one.
In RBA, the contents are considered randomly, but with
respect to their total numbers of requests; a content with higher number of requests will be more likely selected for caching.
In PBA, popular contents, i.e., contents with higher number of requests, will be considered first. For the content under consideration, if the content was not in the cache in the previous time slot, it is downloaded with AoI zero. Otherwise, if AoI cost has reached fifty percent of downloading cost, the content is re-downloaded. Otherwise, the content is kept and the AoI increases by one.

The content popularity distribution is modeled by a ZipF distribution\cite{KarthikeyanShanmugam2013}, i.e., the probability that a user requests the $f$-th content is $\frac{f^{-\gamma}}{\sum_{i \in \mathcal{F}}i^{-\gamma}}$. The popularities of contents are changed randomly across the time slots.
We set $U=600$, $F=200$, and $T=24$ with length of one hour for each time slot \cite{8691020}. The sizes of files are uniformly generated within interval $[1,10]$. The cache capacity is set as $S=\rho \sum_{f \in \mathcal{F}}l_f$. Here, $\rho \in [0,1]$ shows the size of cache in relation to the total size of all contents. The number of requests for each user is randomly generated in $[1,15]$.

The performance results are reported in Figures~$\ref{impact_U}\text{-}\ref{impact_lambda}$.
The deviation from global optimum is bounded by the deviation from the LB, as LB is always less than or equal to the global optimum. We refer to the deviation from LB as optimality gap. The CGA provides solutions within $1\%$ gap from LB and outperforms the conventional algorithms. Figure~\ref{impact_U} shows the impact of $U$. When $U$ increases from $400$ to $800$ the cost nearly linearly increases, however, the optimality gap of algorithms decreases. The reason is that with larger number of users, more contents from the content set are requested by users. As the cache capacity is limited, the only way to get many of requested contents is from the server by all algorithms which leads to a lower optimality gap.

Figure~\ref{impact_F} shows the impact of $F$. Recall that the cache capacity is set to $50\%$ of the total size of the files. For CGA, when $F=50$ the capacity of cache is extremely limited and as $F$ is small, almost all contents will be requested by users. These together imply that many requests need to be satisfied from the server which leads to a high cost. When $F$ increases to $150$, the cost decreases. Because as $F$ increases the cache capacity increases, and CGA is able to efficiently utilize the cache capacity. However when $F$ further increases to $300$, the cost increases. The reason is that even though the capacity increases with $F$ but the diversity of requested contents becomes too large, and consequently some of them need to be satisfied from the server which leads to a higher cost.

Figure~\ref{impact_lambda} shows the impact of $\lambda$. Recall that larger $\lambda$ means higher backhaul load but smaller AoI. From the figure, it can be seen that when $\lambda$ grows, PBA and RBA push down the average AoI of contents to almost zero but incur substantial amount of load on the backhaul. In contrast, the solutions of CGA achieve a much better balance between the backhaul load and AoI of contents with respect to $\lambda$.
Note that the backhaul load and average AoI are normalized to interval $[0,100]$. 


\vspace{-3mm}
\begin{figure}[ht!]
\centering
\includegraphics[scale=0.400]{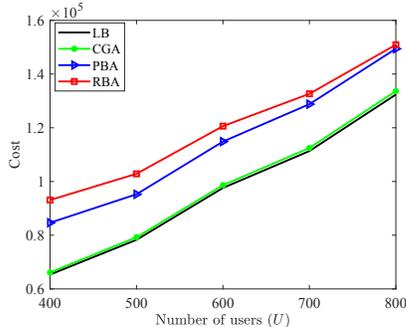}
\begin{center}
\vspace{-6mm}
\caption{Impact of $U$ on cost when $T=24$, $F=200$, $\lambda=0.5$, $\rho=0.5$, $\gamma=0.54$, $c_s=10$,\text{~and~}$c_b=1$.}
\label{impact_U}
\end{center}
\end{figure}
\vspace{-10mm}
\begin{figure}[ht!]
\centering
\includegraphics[scale=0.40]{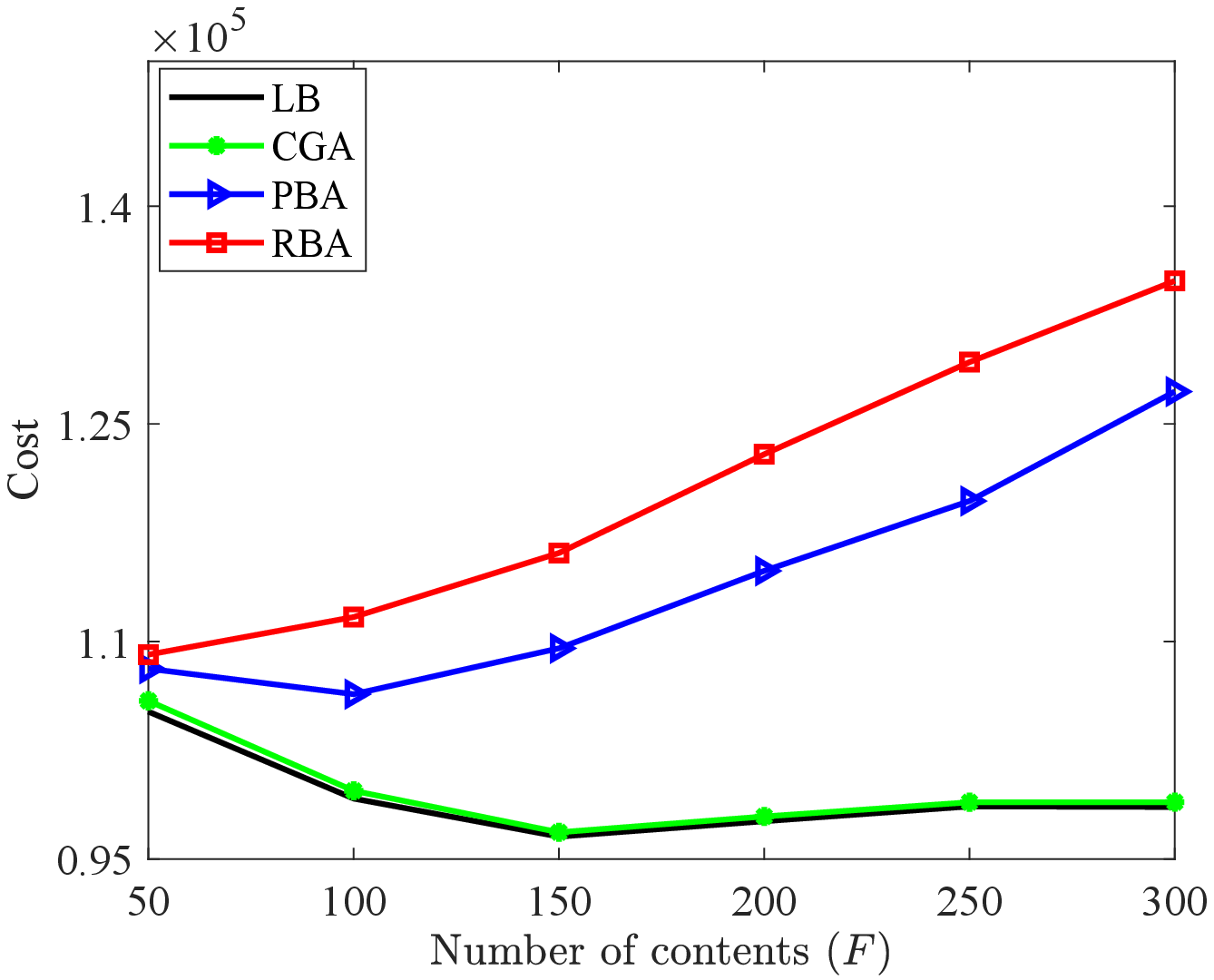}
\begin{center}
\vspace{-6mm}
\caption{Impact of $F$ on cost when $T=24, U=600$, $\lambda=0.5$, $\rho=0.5$, $\gamma=0.54$, $c_s=10$,\text{~and~}$c_b=1$.}
\label{impact_F}
\end{center}
\end{figure}
\vspace{-10mm}
\begin{figure}[ht!]
\centering
\includegraphics[scale=0.40]{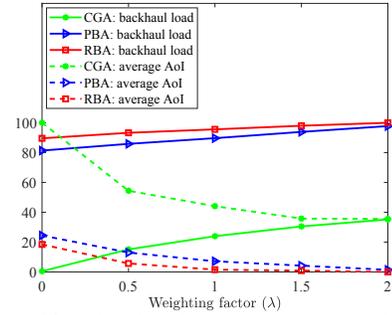}
\begin{center}
\vspace{-6mm}
\caption{Impact of $\lambda$ on backhaul load and average AoI when $T=24, U=600$, $F=200$, $\rho=0.5$, $c_s=10$,\text{~and~}$c_b=1$.}
\label{impact_lambda}
\end{center}
\end{figure}

\section{Conclusions}\label{sec:conclo}
This paper has investigated scheduling of content caching along time where jointly offloading effect and freshness of the contents are accounted for. The problem is formulated as an ILP and $\mathcal{NP}$-hardness of the problem is proved. Next, via a mathematical reformulation, a solution approach based on column generation and a rounding mechanism is developed. Via the joint cost function, it is possible to address the trade-off between the updating and AoI costs. The numerical results show that our algorithm is able to balance between the two costs. Simulation results demonstrated that our solution approach provides near-optimal solutions.

\bibliographystyle{IEEEtran}
\bibliography{IEEEabrv,ForIEEEBib}
\end{document}